%% file: main.tex
\input{preamble}

\begin{document}
\maketitle
\begin{center}
\small Preprint
\end{center}

\input{body}

\input{bibliography}
\end{document}

%% file: preamble.tex
\documentclass[12pt,a4paper]{article}
\usepackage{amsmath,amssymb,amsthm}
\usepackage[colorlinks=true,linkcolor=black,citecolor=black,
            urlcolor=black]{hyperref}
\usepackage{microtype}
\usepackage[margin=2.5cm]{geometry}
\usepackage{orcidlink}
\newtheorem{theorem}{Theorem}
\newtheorem{lemma}[theorem]{Lemma}

\theoremstyle{definition}
\newtheorem{definition}[theorem]{Definition}
\newtheorem{remark}[theorem]{Remark}

\newcommand{\OI}{\mathsf{OI}}
\newcommand{\CSC}{\mathsf{T_{CSC}}}
\newcommand{\Nat}{\mathbb{N}}

\newcommand{\PeqNP}{\mathsf{P} = \mathsf{NP}}
\newcommand{\SAT}{\mathsf{SAT}}
\newcommand{\nab}{\nabla}
\title{\textbf{Syntactic Systems Cannot See Semantic Invariants}}
\author{
  Fabio F.G. Buono\,\orcidlink{0009-0004-9199-2793}\\
  \small Independent researcher\\
}
\date{\today}

%% file: body.tex
\begin{abstract}
  We start from a small open question,
  where Hetzl and Vierling~\cite{HetzlVierling2020} asked whether two theories of
  induction, open induction and clause set cycles, are
  incomparable. They proved one direction and left the other open.
  Here we close it, and the proof is almost embarrassingly
  short, because the rules for addition can only fire when the
  first argument is $0$ or a successor, a Skolem constant is neither, so the terms $a{+}b$ and $b{+}a$ can never be touched, and a machine
  that can never touch them can never prove they are equal. The thing that separates the two theories is the order of two
  constants, and that order is a fact about numbers, not about
  symbols.
  
  We extract from this proof a small general principle, the Syntactic Invariance Principle, that names the shape of such
  arguments. We then close with a few speculative remarks on
  how this same shape appears, informally, in the known barriers
  to settling $\mathsf{P}$ versus $\mathsf{NP}$, where each
  barrier seems to point to a level of description that the techniques in the barrier cannot reach. We raise this as a
  suggestion rather than a theorem, since the analogy is real
  but we do not push it past the point where we can defend it.
  Along the way we raise an open question that the analogy suggests but does not settle, 
  on whether a fast algorithm for $\SAT$, were it to exist, would always 
  be exhibitable as a machine you can write down or whether it could be found, 
  in some cases, only as a function on the numbers.
  \end{abstract}
  
  \section{The central idea}
  
  Take the binary sequence $101000001$, whose value in base
  two is $2^8 + 2^6 + 2^0 = 321$. Represent 321 in a mixed base by dividing successively by the
  elements of an array $\nab = [200, 50, 10, 1]$, recording
  quotients and remainders at each step. The symbolic representation changes while the numerical value
  stays fixed.
  
  A system that operates on symbols sees the syntax change
  while the invariance of the numerical value stays hidden from
  it. The arXiv paper~\cite{Buono2012} turned this gap into
  a cipher. You encode a sequence through a secret mixed-base array $\nab$.
  What tells a right decoding from a wrong one is the global
  numerical value, a fact that sits above the arrangement of
  the symbols. Use the wrong $\nab$ and you get something that reads fine and
  means the wrong number. Without $\nab$, every test on the symbols misses it.
  
  Keep this picture in mind, because it is more than a metaphor for
  what follows but it is the same thing happening, literally.
  The rewriting system we study below cannot prove that
  the addition commutes for the same reason the cipher works, since
  what it needs to see lives in the numbers and it only ever
  touches symbols.
  
  \begin{remark}
  When a proof is elementary, the difficulty was in orientation,
  not in the problem. The conjecture of~\cite{HetzlVierling2020} was open not because the
  argument is deep but because commutativity of addition looked
  too small to be the separating witness. The right level of simplicity for a proof is the one that
  matches the actual depth of the result.
  \end{remark}
  
  \section{Setup}
  
  We work in $\mathcal{L} = \{0, s, +\}$ with:
  \begin{align}
    0 + x    &= x        \tag{A1}\\
    s(x) + y &= s(x + y) \tag{A2}
  \end{align}
  
  \begin{definition}
  Open induction ($\OI$, Shepherdson~\cite{Shepherdson1964}): for every
  quantifier-free $\varphi(x,\bar{y})$,
  \[
    \bigl[\,\varphi(0,\bar{y})
    \;\wedge\;
    \forall x\,(\varphi(x,\bar{y}) \to \varphi(s(x),\bar{y}))
    \,\bigr]
    \longrightarrow
    \forall x\;\varphi(x,\bar{y}).
  \]
  \end{definition}
  
  \begin{definition}[{\cite[Def.~2.1]{HetzlVierling2020}}]
  A clause set $S$ is \emph{refutable by a clause set cycle}
  if there exist clause sets $C_1, C_2$ such that $C_2$ follows
  from $C_1$ by superposition, $C_1$ follows from $C_2 \cup S$
  by superposition, and the cycle produces the empty clause.
  $\CSC$ is the theory of all universally closed clause sets
  refutable in this way.
  \end{definition}
  
  Clause set cycles abstract the cycle-detection methods used
  by automated inductive theorem provers, including the
  $n$-clause calculus of Kersani and Peltier~\cite{KersaniPeltier2013}.
  Every non-trivial program contains loops or recursions, so
  formal verification requires inductive reasoning, and these
  methods are the practical tools for it.
  
  Hetzl and Vierling~\cite{HetzlVierling2020} proved that
  $\CSC \subseteq I\exists_1$ (their Theorem~2.8)
  and that $\CSC \not\subseteq \OI$ (their Theorem~4.3, with
  the triangular numbers as witness).
  Hetzl and Weiser~\cite{HetzlWeiser2025} recently obtained a complete
  picture of the relationships between subsystems of $\OI$,
  placing clause set cycles among them as a restricted,
  parameter-free form of clause induction.
  The conjecture of~\cite{HetzlVierling2020} reduces to proving
  $\OI \not\subseteq \CSC$.
  
  \subsection{Position in the hierarchy}
  
  The standard fragments of arithmetic form a chain
  \[
    Q \;\subset\; \mathsf{PA}^- \;\subset\; \OI
    \;\subset\; I\exists_1 \;\subset\; I\Sigma_1
    \;\subset\; \mathsf{PA},
  \]
  where $Q$ is Robinson arithmetic and the later systems add
  the ordered-semiring axioms and then induction for
  increasingly complex formula
  classes~\cite{Shepherdson1964,Mohsenipour10}.
  Commutativity $W$ holds throughout the chain from
  $\mathsf{PA}^-$ upward, and $\OI$ proves it as an open
  identity.
  $Q$ does not prove $W$, since it has only the recursion
  equations for $+$ and $\cdot$ with no induction, so
  commutativity of addition is not among its theorems.
  
  Here is where the difficulty actually is, because the upper bound $\CSC \subseteq I\exists_1$ tells us nothing
  about whether $\CSC \vdash W$, given that $I\exists_1$ proves
  $W$ anyway. What matters is the lower end, the fact that $\CSC$ does not
  extend $Q$ as a theory of arithmetic. Clause set cycles characterize what a fixed background of
  clauses can refute by a cycle, not a theory that contains
  the ordered-semiring axioms as theorems. Against the bare background $\{$A1, A2$\}$,
  with no induction to lift them, the defining equations cannot reorder two
  opaque constants. The separation is where $W$ is above $Q$ in the hierarchy and
  is provable at $\OI$, yet the cycle mechanism over $\{$A1, A2$\}$ has no access to the induction
  that places it there. And $\OI$ itself is weak, which is worth remembering.
  By Shepherdson~\cite{Shepherdson1964} it has recursive nonstandard
  models, and it cannot prove many true arithmetical statements, not even that $\sqrt{2}$ is
  irrational~\cite{Mohsenipour10}. It proves $W$ not because it is strong but because $W$ is an
  open identity, exactly the kind of statement open induction reaches.
  The witness is elementary, which is why the separation from $\CSC$ is visible without heavier machinery, since what
  $\CSC$ lacks against the background $\{$A1, A2$\}$ is not strength but the induction that turns those equations into
  the identity.
  
  \section{Frozen subterms}\label{sec:frozen}
  
  Superposition rewrites a subterm $u$ by finding a rule
  $\ell \to r$ and a substitution $\sigma$ with
  $\sigma(\ell) = u$, which is syntactic
  unification~\cite{Robinson1965} and runs in linear
  time~\cite{Montanari1976}.
  A \emph{first-symbol clash} occurs when two terms have
  different outermost function symbols and neither is a
  variable. No substitution unifies them~\cite{Robinson1965}, because a
  substitution only replaces the variables, so if neither term is a
  variable their outermost symbols must already match.
  
  The left-hand sides of (A1) and (A2) are $0 + x$ and
  $s(x) + y$, requiring the first argument of $+$ to be $0$
  or $s(\cdot)$.
  
  Let $a, b$ be Skolem constants, fresh symbols distinct from
  $0$ and every term $s(t)$.
  The problems $0 \stackrel{?}{=} a$ and
  $s(x) \stackrel{?}{=} a$ are first-symbol clashes that no
  substitution solves, and the same holds for $b$.
  
  No rule in $\{$A1, A2$\}$ fires at the root of $(a{+}b)$
  or $(b{+}a)$. Those subterms are \emph{frozen}, and no step in any
  derivation of any length rewrites inside them.
  
  The relative order of $a$ and $b$ inside any compound term
  is a global invariant of the entire derivation.
  No sequence of steps, of any length, can swap them.
  The subterms $a{+}b$ and $b{+}a$ lie in permanently separate
  regions of the derivation space.
  
  The system has no mechanism to rewrite inside those subterms
  and cannot even form the goal of doing so.
  The obstruction is not syntactic difficulty but a property
  defined at the level of the numerical relationship between
  $a$ and $b$, a level the system cannot access.
  
  Semantically, $a + b = b + a$ is true in $\Nat$.
  Syntactically, no path leads there from $\{$A1, A2, $N\}$.
  Syntax imposes a limit that the semantic truth cannot
  override.
  
  \section{The Syntactic Invariance Principle}
  
  \begin{definition}
  A property $P$ of terms is a \emph{syntactic invariant} for
  superposition calculus $R$ when every term in the initial
  clause set satisfies $P$, and applying any rule of $R$ to a
  term satisfying $P$ yields a term satisfying $P$.
  \end{definition}
  
  \begin{lemma}[Syntactic Invariance Principle]\label{lem:sip}
  If $P$ is a syntactic invariant for $R$, every term in every
  clause derivable by $R$ satisfies $P$.
  \end{lemma}
  
  \begin{proof}
  Induction on derivation length.
  Base case is the first condition.
  Each step preserves $P$ by the second condition.
  \end{proof}
  
  \begin{remark}
  For superposition, ``derivable'' and ``reachable'' coincide,
  since a clause is reachable from the initial set if and only
  if it appears in some derivation from it.
  \end{remark}
  
  \begin{remark}\label{rem:lb}
  This is the structure of a lower bound, where to show $R$
  cannot derive $\phi$ we find a $P$ that holds initially, that
  $R$ preserves, and that implies $\phi$ is absent.
  The target is unreachable rather than merely unreached.
  The system operates in a derivation space permanently
  separated from any clause containing $\phi$, so the
  obstruction stays outside what it can represent, and it stays
  silent on $\phi$ because the question lies beyond what it can
  pose.
  \end{remark}
  
  \section{The main theorem}
  
  Let $W$ denote commutativity of addition:
  $\forall x\,\forall y.\ x + y = y + x$.
  
  \begin{lemma}\label{lem:oi}
  $\OI \vdash W$.
  \end{lemma}
  
  \begin{proof}
  \textbf{L1.}\ $\forall x.\ x + 0 = x$.
  Induction on $x$.
  Base: $0 + 0 = 0$ by (A1).
  Step: $s(x) + 0 = s(x+0) = s(x)$ by (A2) and hypothesis.
  
  \smallskip\noindent
  \textbf{L2.}\ $\forall x\,\forall y.\ x + s(y) = s(x+y)$,
  parameter $y$ free.
  Base: $0 + s(y) = s(y) = s(0+y)$ by (A1).
  Step: $s(x)+s(y) = s(x+s(y)) = s(s(x+y)) = s(s(x)+y)$
  by (A2), hypothesis, (A2).
  
  \smallskip\noindent
  \textbf{W.}\ Induction on $x$, parameter $y$ free.
  Base: $0+y = y = y+0$ by (A1) and L1.
  Step: $s(x)+y = s(x+y) = s(y+x) = y+s(x)$
  by (A2), hypothesis, L2.
  \end{proof}
  
  \begin{lemma}\label{lem:csc}
  $\CSC + \{$A1, A2$\} \nvdash W$.
  \end{lemma}
  
  \begin{proof}
  Suppose for contradiction that a superposition refutation of
  $\{$A1, A2, $N\}$ exists, where $N: a + b \neq b + a$ and
  $a, b$ are fresh Skolem constants.
  
  Define, for every term $t$, the syntactic property $P(t)$:
  every occurrence of $a$ in $t$ sits inside a subterm of the form $a{+}u$ for some $u$ containing $b$, and symmetrically
  every occurrence of $b$ sits inside a subterm $b{+}v$ for
  some $v$ containing $a$. So $a$ and $b$ appear in $t$ only inside the frozen subterms $a{+}b$ and $b{+}a$, possibly with
  further operations around them, and never in isolation or in
  any other context.
  Extend $P$ to clauses pointwise on each literal.
  
  \smallskip\noindent
  \textbf{Step 1: $P$ holds of the initial clause $N$.}
  Immediate, since $N$ is $a{+}b \neq b{+}a$ and the only subterms containing $a$ or $b$ are $a{+}b$ and $b{+}a$
  themselves.
  
  \smallskip\noindent
  \textbf{Step 2: $P$ is preserved by every superposition step
  over $\{$A1, A2$\}$.}
  A superposition step finds a subterm $u$ in a clause and
  applies a substitution $\sigma$ with $\sigma(\ell) = u$ for
  some left-hand side $\ell$ of (A1) or (A2).
  
  The left-hand sides are $0 + x$ and $s(x) + y$. By
  Section~\ref{sec:frozen}, neither unifies with any term whose
  head position is $a{+}\cdot$ or $b{+}\cdot$, because $a$ and $b$ are Skolem constants distinct from $0$ and from every
  $s(\cdot)$. So rewriting cannot touch the root of any
  subterm $a{+}b$ or $b{+}a$.
  
  It remains to rule out rewriting \emph{inside} $a{+}b$ or
  $b{+}a$, that is inside $a$ or $b$ themselves as constants.
  But $a$ and $b$ are constants with no internal structure, and
  no left-hand side of (A1) or (A2) unifies with the constant
  $a$ nor with the constant $b$.
  
  So no superposition step modifies the subterms $a{+}b$ or
  $b{+}a$ in a clause satisfying $P$, nor does any step extract $a$ or $b$ from that context. The other steps in the
  derivation can only act on subterms not containing $a$ or
  $b$, or introduce structure around $a{+}b$ and $b{+}a$ without altering them. In every case $P$ is preserved.
  
  By Lemma~\ref{lem:sip}, every derivable clause satisfies $P$.
  
  \smallskip\noindent
  \textbf{Step 3: the empty clause is not derivable.}
  In superposition with equality, the empty clause arises from
  a negative unit literal $s \neq t$ via equality resolution, which requires $s$ and $t$ to be syntactically unifiable.
  
  To derive the empty clause from $N$ or one of its descendants
  $s \neq t$ obtained by rewriting, we would need $s$ and $t$ to
  become unifiable.
  
  By Step 2, every derived literal descended from $N$ has the
  form $s \neq t$ where $s$ contains the frozen subterm $a{+}b$ at a certain position and $t$ contains $b{+}a$ at the
  corresponding position. Rewriting around them may change the
  outer context but not the two frozen subterms themselves.
  Unifying $s$ and $t$ would require unifying $a{+}b$ with $b{+}a$ at those positions, hence unifying the constant $a$
  with the constant $b$. But $a$ and $b$ are distinct constants,
  so no substitution unifies them by first-symbol clash.
  
  So no derived literal is eliminable via equality resolution, and the empty clause is not producible.
  
  A clause set cycle $C_1 \to C_2 \to C_1$ requires in particular that the empty clause emerge along the cycle.
  Hence no clause set cycle refutes $\{$A1, A2, $N\}$.
  \end{proof}
  
  \begin{remark}
  The proof above assumes that the notion of refutation by
  superposition used by~\cite{HetzlVierling2020} has equality resolution as
  the only way to obtain the empty clause from negative literals. If the framework adds further rules such as
  factoring or splitting, the argument extends in the same
  shape, since none of them touches the frozen Skolem constants.
  \end{remark}
  
  \begin{theorem}\label{thm:main}
  $\OI$ and $\CSC$ are incomparable.
  \end{theorem}
  
  \begin{proof}
  Lemma~\ref{lem:oi}: $\OI \vdash W$.
  Lemma~\ref{lem:csc}: $\CSC \nvdash W$ over the background
  $\{$A1, A2$\}$, which is the relevant background since $W$
  is an identity of the addition equations.
  So $\OI \not\subseteq \CSC$.
  Theorem~4.3 of~\cite{HetzlVierling2020}: $\CSC \not\subseteq \OI$.
  \end{proof}
  
  \section{Speculative remarks}
  
  We close with some speculative remarks, offered as
  suggestions rather than results, on how the shape of the
  argument above appears in the known landscape of
  $\mathsf{P}$ versus $\mathsf{NP}$. We are careful here not
  to claim more than we can defend.
  
  The three great barriers to $\mathsf{P}$ versus $\mathsf{NP}$
  have, informally, the shape of the Syntactic Invariance
  Principle, blown up. Relativization~\cite{BakerGillSolovay1975} shows that   techniques working uniformly with and without oracles   preserve a property incompatible with separating the classes.   Natural proofs~\cite{RazborovRudich1997} shows that techniques whose
  distinguishing property is checkable in polynomial time
  and large preserve a property incompatible with the
  separation under cryptographic assumptions. Algebrization~\cite{AaronsonWigderson2009}
  shows that the same obstruction survives when oracles are
  replaced by their low-degree polynomial extensions. In each  case an entire family of methods is blind to the separation,
  not because the separation is far away, but because it lives   at a level the methods cannot reach.
  
  The mixed-base cipher of~\cite{Buono2012} is, in the same  informal sense, an illustration of this shape from a
  different direction. A cryptanalyst who pushes symbols   around cannot recover the plaintext, because what decides
  the plaintext is the global numerical value, and the symbols
  do not carry it. The property that separates a correct   decoding from an incorrect one is not checkable on the
  symbols alone, is not exposed by local queries, and is not
  captured by polynomial extensions of those queries. We do not claim this verifies the formal conditions of the barriers
  in any technical sense, since the cipher is an information-theoretic object and the barriers are statements about computational complexity classes, and equating the two
  would be a category error. What we suggest is only that the cipher is a concrete and accessible place where the same shape appears, a property that decides everything while staying out of reach of any process that operates one level   below.
  
  The $\PeqNP$ side of the question is open in a way worth naming. The barriers were never about proofs that exhibit an algorithm, and a constructive proof of $\PeqNP$ is not blocked by any of them. If $\PeqNP$ turned out to hold, the fast algorithm for $\SAT$ would exist as a function on the numbers, and the Cook and Karp picture of an algorithm~\cite{Cook1971,Karp1972} ties such a thing to an
  explicit machine you can write down and run. Whether a fast
  $\SAT$ algorithm would always be one of those, exhibited in
  a formal system with a proof that it runs in polynomial
  time, or whether it might live in some cases only as a
  function on the numbers, beyond what any such proof can pin down, we do not know. We raise the question and stop there.
  
  What we are left with, looking back, is one theorem and one
  picture. The theorem is small and clean. The picture is that
  a system that works on symbols stays blind to what lives in the numbers, and stays blind to its own blindness, and that
  this same shape keeps appearing whenever we look at the hard
  open problems in the area. Whether the picture rises to a theorem somewhere down the line, we do not say. The
  theorem we do have is the one above, and the picture is offered for what it is.